\documentclass{amsart}

\usepackage{amssymb}
\usepackage{amsfonts}
\usepackage{amsmath}
\usepackage{graphicx}%
\usepackage{lastpage}
\usepackage[legalpaper,bookmarks=true,colorlinks=true,linkcolor=blue,citecolor=blue]{hyperref}
\usepackage{fancyhdr}
\usepackage{color}
\usepackage[mathlines]{lineno}
\usepackage{lscape}
\usepackage{epsfig}
\usepackage{natbib}
\usepackage{float}

\newtheorem{theorem}{Theorem}
\theoremstyle{plain}

\newtheorem*{lemma}{Lemma}

\numberwithin{equation}{section}

\begin{document}
\title[Gini Index Variation and Welfare Index]{On the joint distribution of variations of the Gini index and Welfare
indices}
\author{$^{(1)}$ Pape Djiby MERGANE}
\email{pdmergane@ufrsat.org}
\author{$^{(1,2,3)}$Gane Samb LO}
\email{gane-samb.lo@ugb.edu.sn}
\author{$^{(4)}$ Tchilabalo Abozou Kpanzou}
\email{kpanzout@gmail.com}

\begin{abstract}
The aim of this paper is to establish the asymptotic behavior of the mutual influence of the Gini index and the poverty measures by using the Gaussian fields described in Mergane and Lo(2013). The results are given as representation theorems using the Gaussian fields of the unidimensional or the bidimensional functional Brownian bridges. Such representations, when combined with those already available, lead to joint asymptotic distributions with other statistics of interest like growth, welfare and inequality indices  and then, unveil interesting results related to the mutual influence between them. The results are also 	appropriate for studying whether a growth is fair or not, depending on the variation of the inequality measure. Datadriven applications are also available. Although the variances may seem complicated at a first sight, their computations which are needed to get confidence intervals of the indices, are possible with the help of R codes we provide. Beyond the current results, the provided representations are useful in connection with different ones of other statistics.\\

\bigskip
\bigskip \noindent $^{(1)}$ LERSTAD, Gaston Berger University, Saint-Louis,
Senegal.\newline
$^{(2)}$ LSTA, Pierre and Marie Curie University, Paris VI, France.\newline
\noindent $^{(3)}$ AUST - African University of Sciences and Technology,
Abuja, Nigeria\newline
\noindent $^{(4)}$ FAST, Kara University, Togo.\newline

\noindent \textit{Corresponding author}. Gane Samb Lo. Email :
gane-samb.lo@edu.ugb.sn, ganesamblo@ganesamblo.net\\
Permanent address : 1178 Evanston Dr NW T3P 0J9,
Calgary, Alberta, Canada.

\end{abstract}


\maketitle


\Large

\section{INTRODUCTION AND MOTIVATION}

\label{sec1}

In this paper, the asymptotic behavior of the Gini inequality index (1921) is jointly studied with a general class of welfare indices within the frame of unified Gaussian fields both for in a one phase frame (fixed time) and in a two phase frame (variation between two periods). Beyond the results themselves, the obtained asymptotic representations allow future couplings of the studied statistics with other indices. These couplings will lead to joint asymptotic distributions, enabling interesting comparison and influence studies between indices.\\
 
\noindent We begin by a survey on the Gini index, based on historical and recent works, concerning its statistical properties, its  asymptotic distributions and some of its generalizations. In a second step, we will explain the notion of Gaussian fields we mentioned before.\\

\noindent The Gini index (1921) has played and is playing an important role in the measurement of
economic inequality since its development by Corrado Gini in the early 20th
century. Besides, this index is also used in many other disciplines, including Biology (Graczyk, 2007), Astronomy (Lisker,
2008), Environment (Druckman and Jackson, 2008; Groves-Kirkby, 2009).\\

\noindent Various expressions for the Gini index are given by authors such as Davidson (2009), Dorfman (1979), Duclos and Araar (2006). Extended Gini indices are also developed (see e.g., Weymark, 1981; Yitzhaki, 1983; Chakravarty, 1988). Over the years, statistical inference for the Gini index has attracted many researchers. For example, Gastwirth (1972) discussed the estimation of the index from that of the Lorenz curve. Cowell and Flachaire (2007) have developed its influence function and looked at how influenced is its non-parametric estimator to extreme values. Moni (1991) also studied the Gini measure by means of the influence function. On their part, Qin et al. (2010) constructed empirical likelihood confidence intervals for the Gini coefficient and showed that these perform well, but only for large samples. In order to improve inference based on it, Sarno (1998) proposed, in a non-parametric setting, a new stabilizing transformation for the sample Gini coefficient.\\

\noindent Fakoor et al. (2011) considered non-parametric estimators of the Lorenz curve and Gini index based on a sample from the corresponding length-biased distribution, showed that such estimators are strongly consistent for the Gini index, and derived an asymptotic normality for that index. Davidson (2009) developed a reliable standard error for the plug-in estimator of the Gini index and derived an effective bias correction. Mart\'inez-Camblor and Corral (2009) developed results on exact and an asymptotic distribution of the Gini coefficient. Asymptotic distribution of the S-Gini index is derived by Zitikis and Gastwirth (2002), who provided an explicit formula for the asymptotic variance. More on inference for the extended Gini indices can be found in, e.g., Xu (2000) and Barrett and Donald (2000).\\

\noindent But the Gini's index is one of a quite few number of inequality measures that are available in the literature. A considerable number of them has been gathered in a class named Theil-like family and studied jointly with welfare statistics. This study dit not concern the Gini's index nor the new Zenga's (\cite{zenga1}) inequality measure. Because of its great importance, a similar handling for the Gini's measure seems to be highly recommended alongside comparison investigations.\\ 

\noindent As mentioned above, a new approach, that is set to put the asymptotic results of indices
related to welfare and inequality analysis in a unified frame of one
Gaussian field,  was attempted in  \cite{merganelo13}. In that paper, a large
class of inequality measures named as the Theil-like family has been
jointly studied with an other general class of poverty measures known under the name of
General Poverty Index (\textit{GPI}), both with respect to a spatial (horizontal) and a
time (vertical) perspective. Such an approach leads to powerful tools when comparing different indicators or their variation over the space or the time scale. Since the joint asymptotic results are expressed with respect to one common Gaussian process, the method makes easy the comparison of the results for one particular index with those for different indices or statistics using the same frame. Our aim is to offer such representations for the Gini's index and to benefit from them, in order to have insightful relations with the GPI. These representations will be used later in a full study of all available inequality measures. In the coming Subsection \ref{subsec12}, we will give a full description of the probability spaces holding the representations.\\

\noindent Our main results start from the complete description of the asymptotic
representation of the Gini's index in a Gaussian field and in a residual Gaussian process $\beta $\ already introduced and studied in \cite{logs} for the fixed time scheme in Theorem \ref{theoGI}. These
results are extended to the two phase variation scheme in Theorem \ref{theodeltaGI}.
Finally, their combination with available representations, yields successful descriptions of the mutual influence of the Gini's index 
and usual poverty indices including the Sen and Kakawani ones in Theorem \ref{theoR}. Unlike former works on the topic, we
appeal to the Bahadur Representation Theorem (see \cite{bahadur66}) as a tool for handling L-statistis in the lines of \cite{logs}. Datadriven studies are included. But beyond this, the representations will serve in connection with similar ones for different indices of interest.\\  

\noindent We will exclusively limit our study in the field of the welfare
analyis and focus on the Gini's index and the General poverty measure. In future works, extentions of our current results will be extended to 
extension of the Gini's measures : the Generalized Gini, S-Gini, E-Gini. (see Barrett and Donald 2009 \cite{barrett}).\\

\noindent Let us recall that we may and do measure poverty (or richness) with the help of poverty indices $J$ based on
the income variable $X$.  To each income, a poverty line $Z>0$ is associated. This poverty line is defined the minimum income under which an individual is declared as \textit{poor}. Over two periods $s=1$ and $t=2,$ we say that we
have a gain against poverty when $\Delta J(s,t)=J(t)-J(s) \leq 0$, or simply
a growth against poverty. But this variation is not enough to describe the
situation of the population, one must be sure that, meanwhile, the income
did not become more \textit{unequally} distributed, that is the appropriate
inequality coefficient $I$ did not increase. One can achieve this by
studying the ratio $R=\Delta J(s,t)/\Delta I(s,t)$, where $\Delta I(s,t) =
I(t) - I(s)$ denotes the variation of the distribution of the income
variable.\newline

\noindent To make the ideas more precise, let us suppose that we are
monitoring the poverty scene on some population over the period time $[1,2]$
and let $Y=(X^1,X^2)$ be the income variable of that population at periods $1$
and $2$. Let us consider one sample of $n \geq 1$ individuals or households,
and observe the income couple $Y_{j}=(X^{(2)}_{j},X^{(2)}_{j})$, $j=1,...,n$. For
each period $i \in \left\{1,2\right\}$, we also denoted by $X^i_{1,n} \leq
X^i_{2,n} \leq \cdots X^i_{n,n}$ the order statistics. We assume that $X ^ i 
$ is strictly positive, and we compute the poverty measure $J_{n}(i)$ and
the inequality measure $I_{n}(i)$.\\

\noindent For poverty, we consider the Generalized Poverty Index (GPI)
introduced by Lo \textit{et al.} \cite{loGPI} and Lo \cite{loGPI2013} as an
attempt to gather a large class of poverty measures reviewed in Zheng \cite%
{zheng} defined as follows for period $i$,

\begin{equation}
J_{n}(i)=\frac{A(Q_{n}(i),n,Z(i))}{nB(Q_{n}(i))}\sum_{j=1}^{Q_{n}(i)}w(\mu
_{1}n+\mu _{2}Q_{n}(i)-\mu _{3}j+\mu _{4})\text{\ }d\left( \frac{%
Z(i)-X^i_{j,n}}{Z(i)}\right)  \label{gpi01}
\end{equation}

\noindent where $B(Q_n(.))=\sum_{j=1}^{n} w(j),$ $Z(i)$ is the poverty line at time $t=i$, $Q_n(.)$ is the number of poor, $\mu _{1},\mu
_{2},\mu _{3}$ and $\mu _{4}$\ are constants, $A(u,v,s),$ $w(t),$ and $d(y)$
are measurable functions of $(u,v,s) \in \mathbb{N}\times \mathbb{N} \times 
\mathbb{R}^{\ast}_{+},$ $t \in \mathbb{R}^{\ast}_{+},$ and $x \in (0,1).$
By particularizing the functions $A$ and $w$ and by giving fixed values to
the $\mu _{i}^{\prime }s,$ we may find almost all the available indices, as
we will do it later on. \textit{In the sequel, (\ref{gpi01}) will be called
a poverty index (indices in the plural) or simply a poverty measure
according to the economists' terminology.}.\\

\bigskip \noindent This class includes the most popular indices such as those of Sen (%
\cite{sen}), Kakwani (\cite{kakwani}), Shorrocks (\cite{shorrocks}),
Clark-Hemming-Ulph (\cite{chu}), Foster-Greer-Thorbecke (\cite{fgt}), etc.
See Lo (\cite{loGPI2013}) for a review of the GPI. From the works of many
authors (\cite{lo1}, \cite{sl} for instance), $J_{n}(i)$ is an
asymptotically sufficient estimate of the exact poverty measure

\begin{equation}
J(i) = \int_{0}^{Z(i)} L \left(x,F_{(2),i}\right) d\left(\frac{Z(i) - x}{Z(i)}%
\right)\,dF_{(2),i}(x)
\end{equation}

\noindent where $F_{(2),i}$ is the distribution function of $X^{(i)} \,(i=1,2)$, and $L$ is some
weight function.

\bigskip

\noindent As for the inequality measure, we only use the Gini index $(GI)$
which is based on the Lorenz curve (1905). And, for a given date $i \in
\left\{1, 2\right\}$, we denote by

\begin{equation}
GI_n(i) = \frac{1}{\mu_n(i)}\frac{1}{n} \sum_{j=1}^{n} \left(\frac{2 j - 1}{n%
} - 1 \right) X^{(i)}_{j,n}
\end{equation}

\noindent the empirical measure of the Gini index (see \textit{Greselin et
al.}, \cite{greselin09}), and its continuous form is defined as follows

\begin{equation}
GI(i) = \frac{1}{\mu(i)}\int_{0}^{+\infty}F_{(2),i}^{-1}(x)\left(2 F_{(2),i}(x) - 1
\right)\,dF_{(2),i}(x),
\end{equation}

\noindent where $\mu(i) = \mathbb{E}(X^i)$ is the mathematical expectation of $X^i$ and $F_{(2),i}^{-1}$ denotes the generalized inverse of the \textit{cdf} $F_{(2),i}$.\\

\noindent The motivations stated above lead to the study of the behavior of 
\begin{equation*}
(\Delta J_{n}(s,t),\Delta GI_{n}(s,t)),
\end{equation*}
defined for two periods $s<t$, as an estimate of the unknown value of 
\begin{equation*}
(\Delta J(s,t),\Delta GI(s,t)).
\end{equation*}

\bigskip

\noindent Precisely confidence intervals of

\begin{equation*}
R(s,t)=\frac{\Delta J(s,t)}{\Delta GI(s,t)} 
\end{equation*}

\noindent will be an appropriate set of tools for the study of the mutual
influence of the Gini index and the poverty measures.

\bigskip

\noindent To achieve our goal we need a coherent asymptotic theory allowing
the handling of longitudinal data as it is the case here and a stochastic
process approach leading to asymptotic sub-results with the help of the
continuity mapping theorem.

\bigskip

\noindent The rest of the paper is structured as follows. In the rest of this Section \ref{sec1}, we describe the probability space on which the asymptotic representations will take place. In Section \ref%
{ABGI} we provide a study on the asymptotic behavior of the Gini index. Then
in Section \ref{VGI},a complete study of the variation of this index between two given dates is provided. And next, Section \ref{MI}, we treat the
mutual influence of the latter on the Generalized Poverty Index (GPI)
introduced by Lo \textit{et al.} \cite{loGPI} and Lo \cite{loGPI2013}. Section \ref%
{App} is devoted to applications of the theoretical results with datadriven examples. The paper ends with a conclusion in Section \ref%
{CR}.\\

\bigskip The notation used in the paper may be seen as complicated, but knowing the following simple facts may help in making them very comprehensive. The subscript $(1)$ means that we are working un on dimension, where the randoms variables do not have a superscript. In dimension 2, we always have the subscript (2) to main functions : \textit{cdf}'s, copulas, empirical process,etc. When followed by $i$, like $F_{(2),i}$, it refers to a margin. For example $F_{(2),1}$ is the first marginal \textit{cdf} of $F_{(2)}$. Still in dimension 2, any superscript $i=1,2$ refers to the first coordinate of a couple.\\

\subsection{Notations and Probability Space} \label{subsec12} $ $\\

\noindent In this Subsection, we complete the notations we already gave and precise our probability space.\\

\noindent \textbf{Univariate frame}. \noindent We are going to describe the general Gaussian field in which we present our results. Indeed, we use a unified approach when dealing with the asymptotic theories of the welfare statistics. It is based on the Functional Empirical Process (\textit{fep}) and its Functional Brownian Bridge (\textit{fbb}) limit. It is laid out as follows.\\

\noindent When we deal with the asymptotic properties of one statistic or index at a fixed time, we suppose that we have a non-negative random variable of interest which may be the income or the expense $X$ whose probability law on $(\mathbb{R},\mathcal{B}(\mathbb{R}))$, the Borel measurable space on $\mathbb{R}$, is denoted by $\mathbb{P}_{X}.$ We consider the space $\mathcal{F}_{(1)}$ of measurable real-valued functions $f$ defined on $\mathbb{R}$\ such that \begin{equation*}
V_{X}(f)=\int (f-\mathbb{E}_{X}(f))^{2}d\mathbb{P}_{X}=\mathbb{E}(f(X)-\mathbb{E}(f(X))^{2}<+\infty ,
\end{equation*}

\noindent where 
\begin{equation*}
\mathbb{E}_{X}(f)=\mathbb{E}f(X).
\end{equation*}

\bigskip \noindent On this functional space $\mathcal{F}_{(1)},$\ which is endowed with the $L_{2}
$-norm
\begin{equation*}
\left\Vert f\right\Vert _{2}=\left( \int f^{2}d\mathbb{P}_{X}\right) ^{1/2},
\end{equation*}

\noindent we define the Gaussian process $\{\mathbb{G}_{(1)}(f),f\in \mathcal{F}_{(1)}\},$
which is characterized by its variance-covariance function

\begin{equation*}
\Gamma_{(1)}(f,g)=\int^{2}(f-\mathbb{E}_{X}(f))(g-\mathbb{E}_{X}(g))d\mathbb{P%
}_{X},(f,g)\in \mathcal{F}_{(1)}^{2}.
\end{equation*}
 
\noindent This Gaussian process is the asymptotic weak limit of the sequence of functional empirical processes (fep) defined as follows. Let $X_{1},X_{2},...$ be a sequence of independent copies of $X$. For each $n\geq 1,$ we define the functional empirical process associated with $X$ by 
\begin{equation*}
\mathbb{G}_{n,(1)}(f)=\frac{1}{\sqrt{n}}\sum_{i=1}^{n}(f(X_{i})-\mathbb{E}%
f(X_{i})),f\in \mathcal{F}_{(1)},
\end{equation*}

\noindent and denote the integration with respect to the empirical measure by

\begin{equation*}
\mathbb{P}_{n,(1)}(f)=\frac{1}{n} \sum_{i=1}^{n}(f(X_{i}), \ f\in \mathcal{F}_{(1)},
\end{equation*}

\noindent Denote by $\ell^{\infty }(T)$ the space of real-valued bounded functions defined on $T=\mathbb{R}$ equipped with its uniform topology. In the terminology of the weak convergence theory, the sequence of objects $\mathbb{G}_{n,(1)}$  weakly converges to $\mathbb{G}_{(1)}$ in $\ell^{\infty}(\mathbb{R})$, as stochastic processes indexed by  $\mathcal{F}_{(1)}$, whenever it is a Donsker class. The details of this highly elaborated theory may be found in Billingsley \cite{billingsley}, Pollard \cite{pollard}, van der Vaart and Wellner \cite{vaart} and similar sources.\\

 \noindent we only need the convergence in finite distributions which is a simple consequence of the multivariate central limit theorem, as described in Chapter 3 in Lo \cite{wcia-srv-ang}.\\

\noindent We will use the Renyi's representation of the random variable $X_i$'s of interest by means (\textit{cdf}) $F_{(1)}$ as follows
$$
X=_{d}F_{(1)}^{-1}(U),
$$

\noindent where $U$ is a uniform random variable on $(0,1)$, $=_{d}$ stands for the equality in distribution and $F^{-1}_{(1)}$ is the generalized inverse of $F_{(1)}$, defined by

$$
F_{(1)}^{-1}(s)=\inf \{x, F_{(1)}(x)\geq s\}, \ s \in (0,1).
$$

\noindent Based on these representations, we may and do assume that we are on a probability space $(\Omega,\mathcal{A},\mathbb{P})$ holding a sequence of independent $(0,1)$-uniform random variables $U_1$, $U_2$, ..., and the sequence of independent observations of $X$ are given by 

\begin{equation}
X_{1}=F_{(1)}^{-1}(U_1), \ \ X_{2}=F_{(1)}^{-1}(U_2), \ \ etc. \label{repRenyi}
\end{equation}

\noindent \noindent For each $n\geq 1$, the order statistics of $U_1,...,U_n$ and of $X_1,...,X_n$ are denoted respectively by $U_{1,n}\leq \cdots \leq U_{n,n}$ and $X_{1,n}\leq \cdots \leq X_{n,n}$.\\

\noindent \noindent To the sequences of $(U_n)_{n\geq 1}$, we also associate the sequence of real empirical functions

\begin{equation}
\mathbb{U}_{n,(1)}(s)=\frac{1}{n} \#\{i,1\leq i \leq n, \ U_i \leq s\}, \ s\in(0,1) \ n\geq 1 \label{empiricalfunctionU}
\end{equation}

\bigskip \noindent and the the sequence of real uniform quantile functions
\begin{equation}
\mathbb{V}_{n,(1)}(s)=U_{1,n}1_{(0\leq s \leq 1/n)}+\sum_{j=1}^{n}U_{j,n}1_{((i-1)/n\leq s \leq (i/n))}, \ s\in(0,1), \ n\geq 1 \label{quantilefunctionU}
\end{equation}

\noindent and next, the sequence of real uniform empirical processes 
\begin{equation}
\alpha_{n,(1)}(s)=\sqrt{n}(\mathbb{U}_{n,(1)}-s), \ s\in(0,1) \ n\geq 1 \label{empiricalprocessU}
\end{equation}

\bigskip 
\noindent and the sequence of real uniform quantile processes

\begin{equation}
\gamma_{n,(1)}(s)=\sqrt{n}(s-\mathbb{V}_{n,(1)}), \ s\in(0,1) \ n\geq 1. \label{quantileprocessU}
\end{equation}

\bigskip \noindent The same can be done for the sequence $(X_n)_{n\geq 1}$, and we obtain the associated sequence of real empirical procecesses a 

\begin{equation}
\mathbb{G}_{n,r,(1)}(x)=\sqrt{n} \left( \mathbb{F}_{n,(1)}(x)-F_{(1)}(x)\right), \ x\in \mathbb{R}, \ n\geq 1 \label{empiricalprocess}
\end{equation}

\bigskip \noindent where 

\begin{equation}
\mathbb{F}_{n,(1)}(x)=\frac{1}{n} \#\{i,1\leq i \leq n, \ X_i \leq x\}, \ x\in \mathbb{R} \ n\geq 1 \label{empiricalfunction}
\end{equation}

\noindent is the associated sequence of empirical functions. We also have the associated sequence of quantile processes 

\begin{equation}
\mathbb{Q}_{n,(1)}(x)=\sqrt{n} \left( \mathbb{F}^{-1}_{(n),(1)}(s) - F^{-1}(s) \right), \ s\in (0,1), \ n\geq 1 \label{quantileprocesses}
\end{equation}

\noindent where, for $n\geq 1$,

\begin{equation}
\mathbb{F}^{-1}_{n,(1)}(s)=X_{1,n}1_{(0\leq s \leq 1/n)}+\sum_{j=1}^{n}X_{j,n}1_{((i-1)/n\leq s \leq (i/n))}, \ s\in(0,1), \label{quantilefunction}
\end{equation}

\noindent is the associated sequence of quantile processes.\\

\noindent By passing, we recall that $\mathbb{F}^{-1}_{n,(1)}$ is actually the generalized inverse of $\mathbb{F}_{(n),(1)}$ and for the uniform sequence, we have

\begin{equation}
\mathbb{V}_{n,(1)}=\mathbb{U}^{-1}_{n,(1)} \label{invCDFEMP}
\end{equation}

\noindent In virtue of Representation (\ref{repRenyi}), we have the following remarkable relations

\begin{equation}
\mathbb{G}_{n,r,(1)}(x)=\alpha_{n,(1)}(F_{(1)}(x)), \ x\in \mathbb{R} \label{EmpEmpprocess}
\end{equation}

\bigskip \noindent and

\begin{equation}
\mathbb{Q}_{n,(1)}(x)=\sqrt{n}\left( F^{-1}_{(1)}(\mathbb{V}_{n,(1)}(s))- F^{-1}_{(1)}(s)\right) \ s\in (0,1), \ n\geq 1, \label{QQprocess}
\end{equation}

\bigskip \noindent We also have the following relations between the empirical functions and quantile functions

\begin{equation}
\mathbb{F}_{n,(1)}(x)=\mathbb{U}_{n,(1)}(F_{(1)}(x)), \ x\in \mathbb{R} \label{EEFunction}
\end{equation}

\bigskip \noindent and

\begin{equation}
\mathbb{F}^{-1}_{n,(1)}(s)=F^{-1}_{(1)}(\mathbb{V}_{(n),(1)}(s)), \ s\in (0,1), \ n\geq 1. \label{QQFunction}
\end{equation}

\noindent As well, the real and functional empirical processes are related as follows : for $n\geq 1$,

\begin{equation}
\mathbb{G}_{n,r,(1)}(x)=\mathbb{G}_{n,(1)}(f_{x}^{\ast}), \ \alpha_{n,(1)}(s)=\mathbb{G}_{n,(1)}(f_s),  \ s \in (0,1)  \ x \in \mathbb{R}, \label{empiricalprocessRealFonct}
\end{equation}

\bigskip \noindent where for any $x \in \mathbb{R}$, $f_{x}^{\ast}=1_{]-\infty,x]}$ is the indicator function of $]-\infty,x]$ and for $s \in (0,1)$, $f_s=1_{[0,s]}$.\\
 
\bigskip 
\noindent To finish the description, a result of Kiefer-Bahadur (See \cite{bahadur66}) that says that the addition of the sequences of uniform empirical processes and quantiles processes (\ref{empiricalprocessU}) and (\ref{quantileprocessU}) is asymptotically, and uniformly on $[0,1]$, zero in probability, that is

\begin{equation}
\sup_{s\in [0,1]} \left\vert \alpha_{n,(1)}(s)+\gamma_{n,(1)}(s) \right\vert =o_{\mathbb{P}}(1) \text{ as } n\rightarrow +\infty. \label{bahadurRep}
\end{equation}

\noindent This result is a powerful tool to handle the rank statistics when our studied statistics are $L$-statistics.\\

\bigskip \noindent \textbf{Bivariate frame}. As to the bivariate case, we use the Sklar's theorem (See \cite{sklar}). Let us begin to define a copula in $\mathbb{R}^2$ as bivariare probability distribution function $C(u,v)$, $(u,v)\in \mathbb{R}^2$ with support $[0,1]^2$ and with $[0,1]$-uniform margins, that is

$$
C(u,v)=0 \text{ for } (u,v)\in]-\infty,0[\times \mathbb{R} \text{ for } (u,v)\in \mathbb{R} \times ]-\infty,0[.
$$

\noindent Let us denote by $F_{(2)}$ the bivariate distribution function of our random couple $Y=(X^{(1)}, X^{(2)})$ and by $F_{(21)}$ and $F_{(22)}$ its margins, which are the \textit{cdf} of $X^{(1)}$ and $X^{(2)}$ respectively. The Sklar's theorem (\cite{sklar}) says that there exists a copula $C_{(2)}$ such that we have

\begin{equation}
F_{(2)}(x,y)=C_{(2)}(F_{(21)}(x), F_{(22)}(y)), \text{ for any } (x,y)\in \mathbb{R}^2. \label{theoSklar}
\end{equation}

\noindent This copula is unique if the marginal \textit{cdf}'s are continuous. In this paper, we will suppose that the marginal \textit{cdf}'s are continuous and then $C_{(2)}$ is unique andfixed for once. By the Kolmogorov Theorem, there exists a probability space $(\Omega,\mathcal{A},\mathbb{P})$ holding a sequence of independent random couples $(U^{(1)}_i,U^{(2)}_n)$, $n\geq 1$, of common bivariate distribution function $C_{(2)}$. On that space the random couples $(F_{(21)}^{-1}(U^{(1)}_n), \ F_{(22)}^{-1}(U^{(2)}_n))$ are independent and have a common bivariate distribution function equal to 
$C_{(2)}$, since

\begin{eqnarray*}
&&\mathbb{P}(F_{(21)}^{-1}(U^{(1)}_i)\leq x_{1}, \ F_{(22)}^{-1}(U^{(2)}_i)\leq x_{2})\\
&=&\mathbb{P}(U^{(1)}_i \leq F_{(21)}(x_{1}), \  U^{(2)}_i\leq F_{(22)}(x_{2}))\\
&=&C_{(2)}(F_{(21)}(x_{1}), \  F_{(22)}(x_{2}))\\
&=& F_{(2)}(x_{1}, \ x_{2}),
\end{eqnarray*}

\noindent by (\ref{theoSklar}), and where we applied the general formula for generalized inverses functions for a \textit{cdf} : 

$$
F^{-1}(s) \leq y \Leftrightarrow  s \leq F(x), \text{ for } (s,x) \in [0,1]\times \mathbb{R}.
$$
 
\noindent For more on interesting properties of generalized inverses of monotone functions, see \cite{wcia-srv-ang}, Chapter 4.\\

\noindent Based on this remark, we place ourselves on the probability space holding the sequence of independent random couples $(U^{(1)},U^{(2)})$, $(U^{(1)}_n,U^{(2)}_n)$, $n \geq 2$, with common distribution function $C_{(2)}$, and the observations from $Y=(X^{(1)}, X^{(2)})=(F_{(2),1}^{-1}(U^{(1)}),F_{(2),2}^{-1}(U^{(2)}))$, are generated as follows :

\begin{equation} \label{renyidim2}
Y_n=(F_{(21)}^{-1}(U^{(1)}_n),F_{(22)}^{-1}(U^{(2)}_n)), \ n\geq 1.
\end{equation}

\noindent In this setting, we rather use the the bidimensional functional empirical process based
on $\left\{\left((U_{i}^{(1)},U_{i}^{(2)})\right)\right\}_{i=1,\ldots,n}$ and defined by

\begin{equation}
\mathbb{T}_{n,(2)}\left(f\right) = \frac{1}{%
\sqrt{n}}\,\sum_{j=1}^{n}\,\left( f\left((U_{i}^{(1)},U_{(i)}^{2})\right) - \mathbb{P}%
_{\left((U^{(1)},U^{(2)})\right)}\left(f\right)\right), \label{empProcUV}
\end{equation}

\noindent whenever $f$ is a function of $(u,v)\in [0,1]^2$ such that $\mathbb{E}(f(U^{(1)},U^{(2)})^2)$.\\

\noindent For any Donsker class $\mathcal{F}_{(2)}([0,1]^2)$, the stochastic process $\mathbb{T}_{n,(2)}$ converges to a Gaussian process $\mathbb{T}$ with variance-covariance function, for $(f,g) \in \mathcal{F}_{(2)}^{2}([0,1]^2)$, denoted by $\Gamma^\ast\left(f,g\right)$, is given by 
 
\begin{equation}  \label{Gamma}
 \int_{[0,1]^2}\left(f(u,v) - \mathbb{P}_{\left(U^{(1)},U^{(2)}\right)}\left(f\right)\right)\left(g(u,v) - 
\mathbb{P}_{\left(U^{(1)},U^{(2)}\right)}\left(g\right)\right)\,dC(u,v) 
\end{equation}

\noindent where

\begin{equation*}
\mathbb{P}_{\left(U^{(1)},U^{(2)}\right)}\left(f\right)=\mathbb{E}\left(
f\left(U^{(1)},U^{(2)}\right) \right)=\int_{[0,1]^2}\,f(u,v)\,dC(u,v).
\end{equation*}

\noindent Another form of the variance-covariance function $\ref{Gamma}$ is also 

\begin{equation}  \label{gamma2}
\gamma_{(2)}(f,g)= \int_{[0,1]^2}f(s)g(t)\left(C(s,t) - s\,t\right)\,ds\,dt
\end{equation}

\noindent By deciding to use the functional empirical process based on the observations provided by the Copula C, the functional empirical process based on the incomes and defined by

\begin{equation}
\mathbb{G}_{n,(2)}\left(g\right) = \frac{1}{%
\sqrt{n}}\,\sum_{j=1}^{n}\,\left( g\left((X_{j}^{(1)},X_{(j)}^{2})\right) - \mathbb{P}%
_{\left((X^{(1)},X^{(2)})\right)}\left(g\right)\right), \label{empProcXY}
\end{equation}

\noindent for any function $g$, defined on $\mathcal{V}_X^{2}$ such that $\mathbb{E}(g(X^{(1)},X^{(2)})^2)<+\infty$ is not used directly. Instead, by using Representation (\ref{renyidim2}), we have

\begin{equation*} \large
\mathbb{G}_{n,(2)}\left(g\right) = \frac{1}{\sqrt{n}}\,\sum_{j=1}^{n}\,\left( g\left(F^{-1}_{(2),1}(U_{j}^{(1)}),F^{-1}_{(2),1}(U_{j}^{(2)})\right) - 
\mathbb{P}_{\left((U^{1},U^{2})\right)}\left(g((F^{-1}_{(2),1}(.),F^{-1}_{(2),2}(.))\right)\right).
\end{equation*}

\noindent Hence the correspondence between the function $g$ in Formula (\ref{empProcXY}) and $f$ in Formula in (\ref{empProcUV}) is  the following.

\begin{equation}
f(s,t)=g\left(F^{-1}_{(2),1}(s),F^{-1}_{(2),2}(t)\right), \ (s,t) \in [0,1]^2, \label{empProcXYUV}.
\end{equation}

\bigskip \noindent All the needed notation are now complete and will allow the expression of the asymptotic theory we undertake here.\\

\section{The asymptotic behavior of the Gini Index} \label{ABGI} 

\noindent Let $X$ denote the income random variable of one given population with a positive mean $\mu=\mathbb{E}(X)$ and let $\mathcal{V}_X$ denote its support.

\begin{equation}  \label{formgini1}
GI_n = \frac{1}{\mu_n}\left(\frac{1}{n}\sum_{j=1}^{n} \left(\frac{2 j - 1}{n}
-1 \right)X_{j,n}\right).
\end{equation}

\noindent Set

\begin{equation}  \label{formAn1}
A_n = \frac{1}{n}\sum_{j=1}^{n}\frac{j}{n} X_{j,n}.
\end{equation}

\noindent We can write this expression of as

\begin{equation}  \label{formAn2}
A_n = \frac{1}{n}\sum_{j=1}^{n}\mathbb{F}_{n,(1)}(X_j)X_j.
\end{equation}

\noindent Formula (\ref{formgini1}) becomes

\begin{equation}  \label{formgini2}
GI_n = \frac{2 A_n}{\mu_n} - 1 - \frac{1}{n}.
\end{equation}

\bigskip

\noindent Before tackling this study, let us first introduce some notations:

\begin{equation*}
\forall x \in \mathcal{V}_X, h(x) = x F_{(1)}(x),\; I_d(x) = x; 
\end{equation*}
\begin{equation*}
\forall s \in (0,1), \ell(s) = F_{(1)}^{-1}(s), f_{s}(x) = \mathbf{1}_{(0 ,F_{(1)}^{-1}(s))}(x), 
\end{equation*}

\bigskip

\noindent And finally, set for real-valued measurable functions $f$ and $g$ defined on $\mathbb{R}$ 

\begin{equation}  \label{gamma1}
\gamma_1(f,g)=\int_0^1 \int_0^1 f(s) g(t)\left( \min(s,t) - s\,t
\right)\,ds\,dt
\end{equation}

\noindent Now, we have the following theorems for the asymptotic behavior,
the first concerns the statistic $A_n$ and the second concerns that of $GI_n.
$ Let us state first the following lemma of the representation.

\begin{lemma}
\label{lemma1} Define

\begin{equation*}
\beta_{n,(1)}(\ell) = \int_0^1 \ell(s)\mathbb{G}_{n,(1)}(f_s)\,ds. 
\end{equation*}

The, the statistic $A_n$ can be represented as follows 
\begin{equation}
A_n = \mathbb{P}_{n,1}(h) + \frac{1}{\sqrt{n}}\beta_{n,(1)}(\ell) + o_p(n^{-1/2}).
\end{equation}
\end{lemma}

\bigskip

\begin{proof}
By decomposing the equation (\ref{formAn2}), we get

\begin{equation*}
A_n = \frac{1}{n}\sum_{j=1}^{n)}F_{(1)}(X_j)X_j + \frac{1}{n}\sum_{j=1}^{n}%
\left(\mathbb{F}_{n,(1)}(X_j) - F_{(1)}(X_j) \right)X_j.
\end{equation*}

\bigskip

\noindent Let us denote the residual term by

\begin{equation*}
Re_n = \frac{1}{n}\sum_{j=1}^{n}\left(\mathbb{F}_{n,(1)}(X_j) - F_{(1)}(X_j) \right)X_j
\end{equation*}

\noindent then


\begin{equation*}
Re_n = \sum_{j=1}^{n} \int_{\frac{j-1}{n}}^{\frac{j}{n}}\left\{
\mathbb{F}_{n,(1)}(\mathbb{F}_{n,(1)}^{-1}(s)) - F_{(1)}(\mathbb{F}_{n,(1)}^{-1}(s))\right\}\,\mathbb{F}_{n,(1)}^{-1}(s) \,ds,
\end{equation*}

\noindent and so

\begin{equation}  \label{Rn}
Re_n = \int_{0}^{1}\left\{ F_{n,(1)}(F_{n,(1)}^{-1}(s)) -F_{(1)}(F_{n,(1)}^{-1}(s))\right\}\,F_{n,(1)}^{-1}(s) \,ds.
\end{equation}

\bigskip

\noindent By using Formulas (\ref{invCDFEMP}), (\ref{EmpEmpprocess}) and (\ref{QQprocess}), we get
\begin{equation*}
\sqrt{n} Re_n = - \int_{0}^{1} \sqrt{n} \left\{\mathbb{U}_{n,(1)}\left( \mathbb{V}_{n,(1)}(s) \right) - \mathbb{V}_{n,(1)}(s)\right\}\,F_{n,(1)}^{-1}
\left(\mathbb{V}_{n,(1)}(s)\right)\, ds 
\end{equation*}

\begin{equation*}
= - \int_{0}^{1} \sqrt{n} \left( s - \mathbb{V}_{n,(1)}(s)\right)F_{(1)}^{-1}\left(\mathbb{V}_{n,(1)}(s)\right)\, ds
\end{equation*}
\begin{equation*}
- \int_{0}^{1} \sqrt{n} \left( \mathbb{U}_{n,(1)}\left(\mathbb{V}_{n,(1)}(s)\right) -s\right)\,G^{-1}\left(\mathbb{V}_{n,(1)}(s)\right)\, ds. 
\end{equation*}

\bigskip

\noindent From Shorack and Wellner \cite{shorackwellner} (page 585), we have

\begin{equation*}
\sup_{0\leq s\leq 1} \left| \mathbb{U}_{n,(1)}\left(\mathbb{V}_{n,(1)}(s) \right)
-s \right| \leq \frac{1}{n}.
\end{equation*}

\noindent Using the notations $\alpha_{n,(1)}$ and $\gamma_{n,(1)}$, we get

\begin{equation*}
\sqrt{n} Re_n = - \int_{0}^{1} \sqrt{n} \left( s - \mathbb{V}_{n,(1)}(s)\right)F_{(1)}^{-1}\left(\mathbb{V}_{n,(1)}(s)\right)\, ds + o_p(1)
\end{equation*}

\begin{equation}
= - \int_{0}^{1}\gamma_{n,(1)}(s)F_{(1)}^{-1}(s)\,ds + o_p(1) 
\end{equation}

\noindent By using the following Bahadur's representation (See Formula \ref{bahadurRep}) and by applying Formula \ref{empiricalprocessRealFonct}, we get

\begin{equation*}
\sqrt{n} Re_n = \int_{0}^{1} \mathbb{G}_{n,(1)}(f_s)\,\ell(s)\,ds + o_p(1), 
\end{equation*}

\noindent then by identification we get

\begin{equation*}
\beta_{n,(1)}(\ell) = \int_{0}^{1} \mathbb{G}_{n,(1)}(f_s)\,\ell(s)\,ds,
\end{equation*}

\noindent which closes the proof.
\end{proof}

\bigskip \noindent Here is the full representation of the asymptotic distribution of Gini's statistic.

\begin{theorem}
\label{theoA} \noindent Let $\mathbb{P}_X(h^2)$ is finite and the function $%
\ell$ is bounded, then when $n$ tends to $\infty,$ $\sqrt{n}\left(A_n - 
\mathbb{P}_X(h) \right) \rightarrow_d\; \mathbb{A}(h) = \mathbb{G}_{(1)}(h) + \beta_{(1)}(\ell)$, with $\mathbb{A}(h) \rightsquigarrow \mathcal{N}%
\left(0,\sigma^2_A\right)$, 

\begin{equation}
\sigma^2_A = \Gamma(h,h) + \Gamma(\beta_{(1)}(\ell),\beta_{(1)}(\ell)) + 2
\Gamma(h,\beta_{(1)}(\ell)), \label{gammaG}
\end{equation}

\noindent with

\begin{equation}
\Gamma(h,h) = \int \left(h(x)-\mathbb{P}_X\left(h\right)
\right)^2\,dF_{(1)}(x), \label{gamma}
\end{equation}

\begin{equation}
\Gamma(\beta_{(1)}(\ell),\beta_{(1)}(\ell)) = \gamma_1(\ell,\ell),  \label{gamman1},
\end{equation}

\noindent where the function $\gamma_1(.,.)$ is defined in Formula (\ref{gamma1}), and 

\begin{equation}
\Gamma(h,\beta_{(1)}(\ell)) = \int_0^1 \ell(s)\left(\int_{\left(x\leq
F_{(1)}^{-1}(s)\right)} h(x)\,dF_{(1)}(x)\right)\,ds - \left(\mathbb{P}_X(h) \right)^2.
\end{equation}
\end{theorem}

\begin{proof} By using the previous lemma, 
its easy to see that 
\begin{equation*}
\sqrt{n} \left(A_n - \mathbb{P}_X(h) \right) = \mathbb{G}_{n,(1)}(h) + \beta_{n,(1)}(\ell) + o_p(1)
\end{equation*}

\noindent which tends to a centered Gaussian process $\mathbb{A}(h) = \mathbb{G}_{(1)}(h) + \beta_{(1)}(\ell)$.

\bigskip

\noindent Now, let us find the variance of this centered process. We have

\bigskip

\begin{equation*}
\sigma^2_A = \mathbb{E}\left(\left(\mathbb{G}_{(1)}(h) + \beta_{(1)}(\ell)\right)^2\right) 
\end{equation*}

\begin{eqnarray*}
\sigma^2_A &=& \mathbb{E}\left(\mathbb{G}_{(1)}(h)^2 \right) + \mathbb{E}\left(\beta_{(1)}(\ell)^2 \right)+ 2 \mathbb{E}\left(\mathbb{G}_{(1)}(h)\beta_{(1)}(\ell)\right)\\
&\equiv& \Gamma(h,h) + \Gamma(\beta_{(1)}(\ell),\beta_{(1)}(\ell)) + 2\Gamma(h,\beta_{(1)}(\ell)). 
\end{eqnarray*}

\bigskip

\noindent By Equation (\ref{gamma}) of the definition of the covariance
function, we find

\begin{equation*}
\Gamma(h,h) = \int \left(h(x)-\mathbb{P}_X\left(h\right)
\right)^2\,dG(x).
\end{equation*}

\noindent Let us compute now the remaining terms as follows.

\begin{equation*}
\Gamma(\beta_{(1)}(\ell),\beta_{(1)}(\ell)) = \mathbb{E}\left(\int_0^1\int_0^1\ell(s)%
\ell(t) \mathbb{G}_{(1)}(f_s)\mathbb{G}_{(1)}(f_t)\,ds\,dt \right)
\end{equation*}

\noindent which gives, by applying Fubini's Theorem,

\begin{equation*}
\Gamma(\beta_{(1)}(\ell),\beta_{(1)}(\ell)) = \int_0^1\int_0^1 \ell(s)\ell(t)\mathbb{E}%
\left( \mathbb{G}_{(1)}(f_s)\mathbb{G}_{(1)}(f_t)\right)\,ds\,dt.
\end{equation*}

\noindent Since we have  

\begin{equation*}
\mathbb{E}\left( \mathbb{G}_{(1)}(f_s)\mathbb{G}_{(1)}(f_t)\right) = \min\left(s,t
\right) - s\,t,
\end{equation*}

\noindent we get

\begin{equation*}
\Gamma(\beta_{(1)}(\ell),\beta_{(1)}(\ell)) = \int_0^1\int_0^1
\ell(s)\ell(t)\left(\min\left(s,t \right) - s\,t \right)\,ds\,dt.
\end{equation*}

\bigskip

\noindent By Equation (\ref{gamma1}), we have

\begin{equation*}
\Gamma(\beta_{(1)}(\ell),\beta_{(1)}(\ell)) = \gamma_1(\ell,\ell). 
\end{equation*}

\bigskip

\noindent For $\Gamma(h,\beta_{(1)}(\ell)) = \mathbb{E}\left(\mathbb{G}_{(1)}
(h)\beta_{(1)}(\ell)\right),$ we obtain

\begin{equation*}
\Gamma(h,\beta_{(1)}(\ell)) = \int_0^1 \ell(s) m(h,f_s)\,ds
\end{equation*}

\noindent with

\begin{eqnarray*}
m(h,f_s) &=& \mathbb{E}\left(\mathbb{G}_{(1)}(h)\mathbb{G}_{(1)}(f_s)\right)\\
&=&\int_{\left(x\leq F_{(1)}^{-1}(s)\right)} h(x)\,dF_{(1)}(x)- s\,\mathbb{P}_X(h).
\end{eqnarray*}

\noindent Finally, we conclude that

\begin{equation}  \label{gammahb}
\Gamma(h,\beta_{(1)}(\ell)) = \int_0^1 \ell(s)\left(\int_{\left(x\leq
F_{(1)}^{-1}(s)\right)} h(x)\,dF_{(1)}(x)- s\,\mathbb{P}_X(h) \right)\,ds.
\end{equation}

\bigskip

\noindent This completes the proof of Theorem \ref{theoA}.
\end{proof}

\bigskip

\bigskip

\noindent For the last part of this section, let's define the continuous
form of the Gini index as follows 
\begin{equation*}
GI = \frac{2\,A}{\mu} - 1 \, \text{ with } A = \mathbb{P}_X(h).
\end{equation*}

\noindent Then we are able to expose the following Theorem.

\begin{theorem} \label{theoGI} Assume that $\mu \neq 0,$ $\ell$ is bounded, the quantities $%
\mathbb{P}_X(h^2)$ and $\mathbb{P}_X(I_d^2)$ are finite, then the following
assertion holds :

\begin{equation*}
\sqrt{n}\left(GI_n - GI\right) =\frac{2}{\mu}\left(\mathbb{G}_{n,(1)}(h-\frac{A}{\mu}I_d)+\beta_{n,(1)}(\ell)  \right) + o_p(1). 
\end{equation*}

\noindent This quantity tends to a centered Gaussian process with variance $%
\sigma^2_{GI}$ which is giving by

\begin{equation*}
\sigma^2_{GI} = \frac{4}{\mu^2} \left(\sigma^2_A + \frac{A^2}{\mu^2}
\Gamma(I_d,I_d) -\frac{2A}{\mu}\left(\Gamma(h,I_d) + \Gamma(I_d,\beta_{(1)}(\ell))\right) \right) 
\end{equation*}

\noindent where $\sigma^2_A$ is giving in Theorem \ref{theoA}; $\Gamma(I_d,I_d) = \mathbb{E}\left((X - \mu)^2\right)$ is the variance of the
random variable $X$; 

\begin{equation*}
\Gamma(h,I_d) = \int \left(h(x) - \mathbb{P}_X(h)
\right)\left(x - \mu \right)\,dF_{(1)}(x); 
\end{equation*}
\begin{equation*}
\Gamma(I_d,\beta_{(1)}(\ell)) = \int_0^1 \ell(s)\left(\int_{\left(x\leq
F_{(1)}^{-1}(s)\right)} x\,dF_{(1)}(x)\right)\,ds - A\,\mu. 
\end{equation*}
\end{theorem}

\bigskip

\begin{proof}
\begin{equation*}
\sqrt{n}\left(GI_n - GI\right) = 2\left( \frac{\sqrt{n}\left(A_n - A\right)}{%
\mu_n} - \frac{A}{\mu\mu_n}\sqrt{n}\left(\mu_n - \mu\right)\right) - \frac{1%
}{\sqrt{n}} 
\end{equation*}
\begin{equation*}
=\frac{2}{\mu}\left(\mathbb{G}_{n,(1)}(h) +\beta_n(\ell) - \frac{A}{\mu}\mathbb{G}_{n,(1)}(I_d) \right) + o_p(1) 
\end{equation*}

\noindent which tends to a centered Gaussian process $\frac{2}{\mu} \left( 
\mathbb{A} - \frac{A}{\mu}\mathbb{G}_{(1)}(I_d)\right)$. Compute now the
expression of the variance. We get

\begin{equation*}
\sigma^2_{GI} = \mathbb{E}\left(\left(\frac{2}{\mu} \left( \mathbb{A} - 
\frac{A}{\mu}\mathbb{G}_{(1)}(I_d)\right)\right)^2 \right) 
\end{equation*}
\begin{equation*}
= \frac{4}{\mu^2} \left(\sigma^2_A + \frac{A^2}{\mu^2} \Gamma(I_d,I_d) -%
\frac{2A}{\mu}\left(\Gamma(h,I_d) + \Gamma(I_d,
\beta_{(1)}(\ell)
)\right) \right). 
\end{equation*}

\bigskip

\noindent Applying the equation (\ref{gamma}) to the functions $h$ and $I_d,$
we obtain the expression of $\Gamma(h,h),$ $\Gamma(I_d,I_d)$ and $%
\Gamma(h,I_d).$ And by replacing the function $h$ by $I_d$ in equation \ref%
{gammahb} we obtain 
\begin{equation*}
\Gamma(I_d, \beta_{(1)}(\ell)) = \int_0^1 \ell(s)\left(\int_{\left(x\leq
F_{(1)}^{-1}(s)\right)} I_d(x)\,dF_{(1)}(x)- s\,\mathbb{P}_X(I_d) \right)\,ds ,
\end{equation*}

\noindent but, remember that $I_d$ is the identity function and

\begin{equation*}
\mathbb{P}_X(I_d) \int_0^1 s\,\ell(s)\,ds = \mu\,A 
\end{equation*}

\noindent then

\begin{equation*}
\Gamma(I_d,\beta_{\ell}) = \int_0^1 \ell(s)\left(\int_{\left(x\leq
F_{(1)}^{-1}(s)\right)} x\,dF_{(1)}(x)\right)\,ds - A\,\mu. 
\end{equation*}

\noindent This completes the proof of this part.
\end{proof}

\bigskip Let us move to the variation of the Gini's statistics.


\section{Variation of the Gini Index between two dates} \label{VGI}

We fully use the setting described in Subsection \ref{subsec12} regarding the two phase approach. We need to adapt the notation and the results found in Theorems \ref{theoA} and \ref{theoGI} to follow the
consequences of the moving from  Formula (\ref{empProcXY}) to Formula (\ref{empProcUV}) through Formula (\ref{empProcXYUV}). Accordingly, define $\forall$ $(u,v)\, \in\, [0,1]^2$ and $%
\forall\, j\,=\,1,\,2:$

\begin{equation*}
\ell_j(u) = F^{-1}_{(2),j}(s),\; h_j(u) = u\, \ell_j(u); \; L_j(u) = \frac{2}{%
\mu(j)} \ell_j(u); 
\end{equation*}
\begin{equation*}
\tilde{f}_j(u,v) = \ell_j \circ \pi_j(u,v), \; f_{j,h}(u,v) = h_j \circ
\pi_j(u,v), 
\end{equation*}
\begin{equation*}
f_{j,s}(u,v) = \pi_i\left(\mathbf{1}_{(0\leq s)}(u),\mathbf{1}_{(0\leq
s)}(v)\right) 
\end{equation*}

\noindent

\noindent where $\pi_j$ is the $j^{th}$ projection;

\begin{equation*}
F^\ast_{j,h}(u,v) = \frac{2}{\mu(j)} f_{j,h}(u,v); 
\end{equation*}
\begin{equation*}
F^\ast(u,v) = F^\ast_{2,h}(u,v) - F^\ast_{1,h}(u,v); 
\end{equation*}
\begin{equation*}
\tilde{F}^\ast(u,v) = 2\left(\frac{A(2)}{\mu(2)^2}\tilde{f}_2(u,v) - \frac{%
A(1)}{\mu(1)^2}\tilde{f}_1(u,v)\right). 
\end{equation*}

\noindent Let 
\begin{equation*}
\beta^\ast_{n,(2)}(L) = \int_[0,1] \left(L_2(s) \mathbb{G}_{n,(2)}(f_{2,s}) - L_1(s) 
\mathbb{G}_{n,(2)}(f_{1,s}) \right)\,ds
\end{equation*}

\noindent be the bidimensional residual process.\\

\noindent We can now expose our main theorem which concerns the variation of
the Gini Index.

\bigskip

\begin{theorem} \label{theodeltaGI} Assume that, for all $j = 1,2,$ $\mu(j)$ is finite and
not null; $L_j$ is bounded; the functions $f_{j,s},$ $f_{j,h},$ $\tilde{f}%
_{j},$ $F^\ast$ and $\tilde{F}^\ast$ are square integrable, then we have
the following convergence in distribution as $n$ tends to infinity: 
\begin{equation*}
\sqrt{n}\left( \Delta GI_n(1,2) - \Delta GI(1,2) \right) \rightarrow_d 
\mathcal{G}_{\Delta GI} \rightsquigarrow \mathcal{N}\left(0,\sigma^2_{\Delta
GI} \right)
\end{equation*}

\noindent with

\begin{equation*}
\sigma^2_{\Delta GI} = \Gamma^\ast\left(F^\ast,F^\ast \right) +
\Gamma^\ast\left(\tilde{F}^\ast,\tilde{F}^\ast \right) +
\Gamma^\ast\left(\beta^\ast_L,\beta^\ast_L\right) 
\end{equation*}
\begin{equation*}
-2 \left( \Gamma^\ast\left(F^\ast,\tilde{F}^\ast \right) +
\Gamma^\ast\left(\tilde{F}^\ast,\beta^\ast_L \right) -
\Gamma^\ast\left(F^\ast,\beta^\ast_L \right) \right) 
\end{equation*}

\noindent where

\begin{equation*}
\Gamma^\ast\left( \beta^\ast_L,\beta^\ast_L\right) = \gamma_1(L_1,L_1) +
\gamma_1(L_2,L_2) - 2 \gamma_2(L_1,L_2); 
\end{equation*}

\begin{equation*}
\Gamma^\ast\left(F^\ast,\beta^\ast_L \right) = 
\int_{[0,1]}\left\{L_2(s)\int_{[0,1]\times(0,s)} F^{\ast}(u,v) dC(u,v) \right\}\,ds
\end{equation*}
\begin{equation*}
- \int_{[0,1]}\left\{ \L _1(s)\int_{(0,s)\times [0,1]} F^{\ast}(u,v) dC(u,v)
\right\} ds 
\end{equation*}
\begin{equation*}
-2\, \left(\frac{A(2)}{\mu(2)} - \frac{A(1)}{\mu(1)} \right) \, \mathbb{P}%
_{(U^{(1)},U^{(2)})}\left(F^{\ast}\right)
\end{equation*}

\noindent $\Gamma^\ast\left(\tilde{F}^\ast,\beta^\ast_L \right)$ is
obtained by replacing ${F}^\ast$ by $\tilde{F}^\ast$ in the previous
expression. And we get the covariances of $\Gamma^\ast\left(F^\ast,F^\ast
\right)$, $\Gamma^\ast\left(\tilde{F}^\ast,\tilde{F}^\ast \right)$ and $%
\Gamma^\ast\left(F^\ast,\tilde{F}^\ast \right)$ by the equation (\ref%
{Gamma}).

\bigskip

\begin{proof}
We get 
\begin{equation*}
\sqrt{n}\left( \Delta GI_n(1,2) - \Delta GI(1,2) \right)
\end{equation*}
\begin{equation*}
= \frac{2}{\mu(2)}\left( \mathbb{T}_{n,(2)}\left(f_{2,h}\right) +
\beta^{\ast}_{n,(2)}\left(\ell_2\right) - \frac{A(2)}{\mu(2)}\mathbb{T}_{n,(2)}\left(%
\tilde{f}_2\right)\right) 
\end{equation*}

\begin{equation*}
-\frac{2}{\mu(1)}\left( \mathbb{T}_{n,(2)}\left(f_{1,h}\right) +
\beta^{\ast}_{n,(2)}\left(\ell_1\right) - \frac{A(1)}{\mu(1)}\mathbb{T}_{n,(2)}\left(%
\tilde{f}_1\right)\right) + o_p(1) 
\end{equation*}

\begin{equation*}
= 2\mathbb{T}_{n,(2)}\left(\frac{f_{2,h}}{\mu(2)}-\frac{f_{1,h}}{\mu(1)} \right) +
2\beta^{\ast}_{n,(2)}\left( \frac{\ell_2}{\mu(2)}-\frac{\ell_1}{\mu(1)}\right)
\end{equation*}
\begin{equation*}
-2\mathbb{T}_{n,(2)}\left( \frac{A(2)}{\mu^2(2)}\tilde{f}_2-\frac{A(1)}{\mu^2(1)}%
\tilde{f}_1 \right) + o_p(1).
\end{equation*}

\bigskip

\noindent We find the next expression

\begin{equation*}
\sqrt{n}\left( \Delta GI_n(1,2) - \Delta GI(1,2) \right) = \mathbb{T}_{n,(2)}\left(F^{\ast}\right) + \beta^{\ast}_{n,(2)}\left(L\right)
\end{equation*}
\begin{equation*}
- \mathbb{T}_{n,(2)}\left(\tilde{F}^{\ast}\right) + o_p(1)
\end{equation*}

\begin{equation*}
\rightarrow_d \mathcal{T}_{\Delta GI} = \mathbb{T}_{(2)}\left(F^{\ast}\right) +
\beta^\ast\left(L\right) - \mathbb{T}_{(2)}\left(\tilde{F}^{\ast}\right)
\rightsquigarrow \mathcal{N}\left(0,\sigma^2_{\Delta GI} \right). 
\end{equation*}

\bigskip

\noindent Now, we search the expression of the variance.

\begin{equation*}
\sigma^2_{\Delta GI} = \mathbb{E}\left(\left(\mathbb{T}_{(2)}(F^\ast) +
\beta^\ast(L)\right)^2 +\mathbb{T}_{(2)}\left(\tilde{F}^{\ast}\right)^2\right) 
\end{equation*}
\begin{equation*}
-2 \mathbb{E}\left( \left(\mathbb{T}_{(2)}(F^\ast) + \beta^\ast(L)\right) 
\mathbb{T}_{(2)}\left(\tilde{F}^{\ast}\right)\right) 
\end{equation*}
\begin{equation*}
= \mathbb{E}\left(\mathbb{T}_{(2)}(F^\ast)^2\right) + \mathbb{E}%
\left(\beta^\ast(L)^2\right) + 2 \mathbb{E}\left(\mathbb{T}_{(2)}
(F^\ast)\beta^\ast(L)\right)
\end{equation*}
\begin{equation*}
+ \mathbb{E}\left(\mathbb{T}_{(2)}(\tilde{F}^\ast)^2\right) - 2 \mathbb{E}\left(%
\mathbb{T}_{(2)}(F^\ast) \mathbb{T}_{(2)}(\tilde{F}^\ast) \right) - 2 \mathbb{E}\left(%
\mathbb{T}_{(2)}(\tilde{F}^\ast) \beta^\ast(L)\right).
\end{equation*}

\bigskip

\noindent Then

\begin{equation*}
\sigma^2_{\Delta GI} = \Gamma^\ast\left(F^\ast,F^\ast \right) +
\Gamma^\ast\left(\tilde{F}^\ast,\tilde{F}^\ast \right) +
\Gamma^\ast\left(\beta^\ast_L,\beta^\ast_L\right) 
\end{equation*}
\begin{equation*}
-2 \left( \Gamma^\ast\left(F^\ast,\tilde{F}^\ast \right) +
\Gamma^\ast\left(\tilde{F}^\ast,\beta^\ast_L \right) -
\Gamma^\ast\left(F^\ast,\beta^\ast_L \right) \right). 
\end{equation*}

\bigskip

\noindent And we get the expressions of the covariances of $%
\Gamma^\ast\left(F^\ast,F^\ast \right)$, $\Gamma^\ast\left(\tilde{F}%
^\ast,\tilde{F}^\ast \right)$ and $\Gamma^\ast\left(F^\ast,\tilde{F}%
^\ast \right)$ by the equation (\ref{Gamma}). For the tree rest, let's find
them.

\bigskip

\noindent Firstly, compute $\Gamma^\ast\left(
\beta^\ast_L,\beta^\ast_L\right).$

\begin{equation*}
\Gamma^\ast\left( \beta^\ast_L,\beta^\ast_L\right) = \mathbb{E}%
\left(\beta^\ast(L)^2\right) 
\end{equation*}
\begin{equation*}
= \mathbb{E}\left(\int_D \left(L_2(s)\mathbb{T}_{(2)}(f_{2,s}) - L_1(s)\mathbb{T}_{(2)}(f_{1,s}) \right)\,ds\right)^2 
\end{equation*}
\begin{equation*}
=\mathbb{E}\left(\int_{[0,1]^2} \left(L_2(s)\mathbb{T}_{(2)}(f_{2,s}) - L_1(s)\mathbb{G%
}(f_{1,s}) \right) \left(L_2(t)\mathbb{T}_{(2)}(f_{2,t}) - L_1(t)\mathbb{T}_{(2)}
(f_{1,t})\right)\,ds\,dt\right) 
\end{equation*}
\begin{equation*}
= \int_{[0,1]^2}L_2(s)L_2(t)\Gamma^\ast\left(f_{2,s},f_{2,t}\right)\,ds\,dt +
\int_{[0,1]^2}L_1(s)L_1(t)\Gamma^\ast\left(f_{1,s},f_{1,t}\right)\,ds\,dt 
\end{equation*}
\begin{equation*}
- \int_{[0,1]^2}L_1(t)L_2(s)\Gamma^\ast\left(f_{1,t},f_{2,s}\right)\,ds\,dt -
\int_{[0,1]^2}L_1(s)L_2(t)\Gamma^\ast\left(f_{1,s},f_{2,t}\right)\,ds\,dt. 
\end{equation*}

\bigskip

\noindent But, we have

\begin{equation*}
\int_{[0,1]^2}L_1(t)L_2(s)\Gamma^\ast\left(f_{1,t},f_{2,s}\right)\,ds\,dt =
\int_{[0,1]^2}L_1(s)L_2(t)\Gamma^\ast\left(f_{1,s},f_{2,t}\right)\,ds\,dt
\end{equation*}

\bigskip

\noindent then

\begin{equation*}
\Gamma^\ast\left( \beta^\ast_L,\beta^\ast_L\right) =
\int_{[0,1]^2}L_1(s)L_1(t)\Gamma^\ast\left(f_{1,s},f_{1,t}\right)\,ds\,dt
\end{equation*}
\begin{equation*}
+ \int_{[0,1]^2}L_2(s)L_2(t)\Gamma^\ast\left(f_{2,s},f_{2,t}\right)\,ds\,dt 
\end{equation*}
\begin{equation*}
- 2 \int_{[0,1]^2}L_1(t)L_2(s)\Gamma^\ast\left(f_{1,t},f_{2,s}\right)\,ds\,dt 
\end{equation*}
\begin{equation*}
\equiv \gamma_1(L_1,L_1) + \gamma_1(L_2,L_2) - 2 \gamma_2(L_1,L_2). 
\end{equation*}

\bigskip

\noindent But

\begin{equation*}
\Gamma^\ast\left(f_{1,s},f_{1,t}\right) = \mathbb{E}\left(\mathbf{1}%
_{(0,s)}\left(U\right)\mathbf{1}_{(0,t)}\left(U\right) \right) - st =
\min(s,t) - s\,t, 
\end{equation*}
\begin{equation*}
\Gamma^\ast\left(f_{2,s},f_{2,t}\right) = \mathbb{E}\left(\mathbf{1}%
_{(0,s)}\left(U^{(2)}\right)\mathbf{1}_{(0,t)}\left(U^{(2)}\right) \right) - st =
\min(s,t) - s\,t, 
\end{equation*}

\bigskip

\noindent and

\begin{equation*}
\Gamma^\ast\left(f_{1,s},f_{2,t}\right) = \mathbb{E}\left(\mathbf{1}%
_{(0,s)}\left(U\right)\mathbf{1}_{(0,t)}\left(U^{(2)}\right) \right) - st = C(s,t)
- s t. 
\end{equation*}

\bigskip

\noindent Then by identification we get

\begin{equation*}
\gamma_1(L_1,L_1)= \int_{[0,1]^2}L_1(s)L_1(t)\left(\min(s,t) -
s\,t\right)\,ds\,dt; 
\end{equation*}

\begin{equation*}
\gamma_1(L_2,L_2)= \int_{[0,1]^2}L_2(s)L_2(t)\left(\min(s,t) -
s\,t\right)\,ds\,dt; 
\end{equation*}

\begin{equation*}
\gamma_2(L_1,L_2)= \int_{[0,1]^2}L_1(s)L_2(t)\left(C(s,t) - s\,t\right)\,ds\,dt. 
\end{equation*}

\bigskip

\noindent Secondly, let's find the expression of $\Gamma^\ast\left(F^\ast,%
\beta^\ast_L \right).$

\begin{equation*}
\Gamma^\ast\left(F^\ast,\beta^\ast_L \right) = \mathbb{E}\left(\mathbb{T}_{(2)}(F^\ast)\beta^\ast_L \right) 
\end{equation*}
\begin{equation*}
= \int_D L_2(s)\Gamma^\ast\left(F^\ast,f_{2,s} \right)\,ds - \int_D
L_1(s)\Gamma^\ast\left(F^\ast,f_{1,s} \right)\,ds 
\end{equation*}

\noindent or

\begin{equation*}
\Gamma^\ast\left(F^\ast,f_{2,s} \right) = \mathbb{E}\left(F^\ast(U^{(1)},U^{(2)})
f_{2,s}(U^{(1)},U^{(2)}) \right) - \mathbb{P}_{(U^{(1)},U^{(2)})}\left(F^\ast\right)\mathbb{P}%
_{(U^{(1)},U^{(2)})}(f_{2,s}) 
\end{equation*}

\begin{equation*}
\mathbb{P}_{(U^{(1)},U^{(2)})}(f_{2,s}) = \int_{[0,1]^2}f_{2,s}(u,v)\,dC(u,v) 
\end{equation*}
\begin{equation*}
= \int_{[0,1]^2} \mathbf{1}_{(0,s)}(v)\,dC(u,v) = C(1,s) = s. 
\end{equation*}

\begin{equation*}
\mathbb{E}\left(F^\ast(U^{(1)},U^{(2)}) f_{2,s}(U^{(1)},U^{(2)} \right) = \mathbb{E}%
\left(F^\ast(U^{(1)},U^{(2)})\mathbf{1}_{(0,s)}(U^{(2)}) \right) 
\end{equation*}
\begin{equation*}
=\int_{[0,1]\times(0,s)} F^\ast(u,v)\,dC(u,v), 
\end{equation*}

\bigskip

\noindent and so, we arrive at

\begin{equation*}
\Gamma^\ast\left(F^\ast,f_{2,s} \right) = \int_{[0,1]\times(0,s)}
F^\ast(u,v)\,dC(u,v) - s\,\mathbb{P}_{(U^{(1)},U^{(2)})}\left(F^\ast\right). 
\end{equation*}

\bigskip

\noindent Similarly, we get

\begin{equation*}
\Gamma^\ast\left(F^\ast,f_{1,s} \right) = \int_{(0,s)\times [0,1]}
F^\ast(u,v)\,dC(u,v) - s\,\mathbb{P}_{(U^{(1)},U^{(2)})}\left(F^\ast\right). 
\end{equation*}

\bigskip

\noindent Therefore, we have

\begin{equation*}
\Gamma^\ast\left(F^\ast,\beta^\ast_L \right) = 
\int_{[0,1]}\left\{L_2(s)\int_{[0,1]\times(0,s)} F^{\ast}(u,v) dC(u,v)\right\}\,ds 
\end{equation*}
\begin{equation*}
- \int_{[0,1]}\left\{\L _1(s)\int_{(0,s)\times [0,1]} F^{\ast}(u,v) dC(u,v)
\right\} ds 
\end{equation*}
\begin{equation*}
-\mathbb{P}_{(U^{(1)},U^{(2)})}\left(F^{\ast}\right)\int_{[0,1]} s(L_2(s)-L_1(s))ds, 
\end{equation*}

\bigskip

\noindent and

\begin{equation*}
\int_{[0,1]} s(L_2(s)-L_1(s))ds = 2\, \int_D \left(\frac{s\,\ell_2(s)}{\mu(2)} - 
\frac{s\,\ell_1(s)}{\mu(1)} \right)
\end{equation*}
\begin{equation*}
=2 \left(\frac{A(2)}{\mu(2)} - \frac{A(1)}{\mu(1)} \right). 
\end{equation*}

\bigskip

\noindent Finally we get the expression of $\Gamma^\ast\left(\tilde{F}%
^\ast,\beta^\ast_L \right)$ by the same way. This achieves the proof of
this theorem.
\end{proof}
\end{theorem}



\section{Mutual influence with the GPI}

\label{MI}

\subsection{Remaind on the GPI}

We consider a class of poverty measures called the Generalized Poverty Index
(GPI) introduced by Lo \textit{and al.} \cite{loGPI} as an attempt to gather
a large class of poverty measures reviewed in Zheng \cite{zheng}. This class
includes the most popular indices such as those of Sen (\cite{sen}), Kakwani
(\cite{kakwani}), Shorrocks (\cite{shorrocks}), Clark-Hemming-Ulph (\cite%
{chu}), Foster-Greer-Thorbecke (\cite{fgt}), etc. See Lo (\cite{loGPI2013})
for a review of the GPI.

\bigskip

\noindent For the variation of the GPI, we need the functions $g_i$ and $%
\nu_i$ provided by the theorem of the representation of the GPI \cite{losall}%
. Put accordingly with these functions:

\begin{equation*}
g_i(x) = c\left(F_{(2),i}(x)\right)q_i(x)\;\mathrm{\ and }\;\nu_i(s)
=c^{\prime}(s)q_i\left(F_{(2),i}^{-1}(s)\right).
\end{equation*}

\bigskip

\noindent We define for all $(u,v)\in [0,1]^2$

\begin{equation*}
f_{i,s}(u,v) = \pi_i\left( \mathbf{1}_{(o,s)}(u),\mathbf{1}%
_{(o,s)}(v)\right),
\end{equation*}
\begin{equation*}
F^{\ast}_{i,J}(u,v)=g_i\circ \tilde{f}_i(u,v)= g_i\circ F_{(2),i}^{-1}\circ
\pi_i(u,v),
\end{equation*}
\noindent and 
\begin{equation*}
F^{\ast}_{J}(u,v) = F^{\ast}_{2,J}(u,v) - F^{\ast}_{1,J}(u,v).
\end{equation*}

\noindent And denote the residual term for the \textit{GPI} by

\begin{equation*}
\beta^{\ast}_{(2)}(\nu) = \int_{[0,1]} \left(\mathbb{T}_{(2)}\left(f_{2,s}\right)\nu_2(s) - 
\mathbb{T}_{(2)}\left(f_{1,s}\right)\nu_1(s)\right)ds.
\end{equation*}

\bigskip

\begin{theorem}
\label{theodeltaGPI} Let $\mu(i)$ finite for $i=1,2$. Suppose that $\mathbb{P%
}_{\left(U^{(1)},U^{(2)}\right)}\left(\left(f_{1,s}\right)^2\right),$ $\mathbb{P}%
_{\left(U^{(1)},U^{(2)}\right)}\left(\left(f_{2,s}\right)^2\right)$ and $\mathbb{P}%
_{\left(U^{(1)},U^{(2)}\right)}\left({F^{\ast}_{J}}^2\right)$ are finite.\\

\noindent Then $\sqrt{n}\left(\Delta J_n(1,2) - \Delta J(1,2)\right)$ converges to $%
\mathcal{G}_{\Delta GPI} = \mathbb{T}_{(2)}\left(F^{\ast}_{J}\right) +
\beta^\ast(\nu)$

\noindent which is a centered Gaussian process of variance-covariance
function:

\begin{equation*}
\sigma^2_{\Delta GPI} = \Gamma^\ast\left(F^\ast_J,F^\ast_J\right) +
\Gamma^\ast\left(\beta^\ast_{\nu},\beta^\ast_{\nu} \right) + 2\,
\Gamma^\ast\left(F^\ast,\beta^\ast_{\nu} \right)
\end{equation*}

\noindent where

\begin{equation*}
\Gamma^\ast\left(F^\ast_J,F^\ast_J\right) =
\int_{[0,1]^2}\left(F^{\ast}_{J}(u,v) - \mathbb{P}_{\left(U^{(1)},U^{(2)}\right)}\left(F^{%
\ast}_{J}\right)\right)^2\,dC(u,v);
\end{equation*}

\begin{equation*}
\Gamma^\ast\left(\beta^\ast_{\nu},\beta^\ast_{\nu} \right) =
\gamma_1(\nu_1,\nu_1) - 2\,\gamma_2(\nu_1,\nu_2) + \gamma_1(\nu_2,\nu_2)
\end{equation*}

\noindent with the covariance functions $\gamma_1(.,.)$ and $\gamma_2(.,.)$
are respectively defined in Equation (\ref{gamma1}) and Equation (\ref%
{gamma2});

\noindent and

\begin{equation*}
\Gamma^\ast\left(F^\ast,\beta^\ast_{\nu} \right) = \int_{[0,1]}\left\{
\nu_2(s) \int_{[0,1]\times (0,s)} F^{\ast}_J(u,v)\,dC(u,v)\right\}\,ds
\end{equation*}
\begin{equation*}
- \int_{[0,1]}\left\{\nu_1(s) \int_{(0,s)\times [0,1]}
F^{\ast}_J(u,v)\,dC(u,v)\right\}\,ds
\end{equation*}
\begin{equation*}
- \mathbb{P}_{\left(U^{(1)},U^{(2)}\right)}\left(F^{\ast}_{J}\right)\, \int_{[0,1]}
s\left(\nu_2(s) - \nu_1(s)\right)\,ds.
\end{equation*}
\end{theorem}

\begin{proof}
See Mergane and Lo (\cite{merganelo13}).
\end{proof}

\bigskip


\noindent We are now able to stable our main results.

\bigskip

\subsection{Mutual influence}

\noindent Let 
\begin{equation*}
R = \frac{\Delta J(1,2)}{\Delta GI(1,2)},\; a = \frac{1}{\Delta GI(1,2)}\: 
\text{ and } \: b = \frac{\Delta J(1,2)}{\left(\Delta GI(1,2)\right)^2}.
\end{equation*}

\begin{theorem} \label{theoR} Supposing that the above mentioned hypotheses are true, then

\begin{equation*}
\left(\sqrt{n}\left(\Delta J_n(1,2) - \Delta J(1,2) \right) , \sqrt{n}%
\left(\Delta GI_n(1,2) - \Delta GI(1,2) \right) \right)^t \rightarrow_d 
\mathcal{N}_2\left(0,\Sigma\right),
\end{equation*}
with 
\begin{equation*}
\Sigma = \left(%
\begin{array}{ccc}
\sigma^2_{\Delta GPI} &  & \sigma_{\Delta GPI,\Delta GI} \\ 
&  &  \\ 
\sigma_{\Delta GPI,\Delta GI} &  & \sigma^2_{\Delta GI} \\ 
&  & 
\end{array}%
\right) 
\end{equation*}

\noindent where

\begin{equation*}
\sigma_{\Delta GPI,\Delta GI} = \Gamma^\ast\left(F^\ast_J,F^\ast\right) +
\Gamma^\ast\left(F^\ast_J,\beta^\ast_L\right) -
\Gamma^\ast\left(F^\ast_J,\tilde{F}^\ast\right) 
\end{equation*}
\begin{equation*}
+ \Gamma^\ast\left(F^\ast,\beta^\ast_{\nu}\right) +
\Gamma^\ast\left(\beta^\ast_{L},\beta^\ast_{\nu}\right) -
\Gamma^\ast\left(\tilde{F}^\ast,\beta^\ast_{\nu}\right) 
\end{equation*}

\bigskip

\noindent with

\begin{equation*}
\Gamma^\ast\left(\beta^\ast_{L},\beta^\ast_{\nu}\right) =
\gamma_1(L_1,\nu_1) + \gamma_1(L_2,\nu_2) - \gamma_2(L_1,\nu_2) -
\gamma_2(L_2,\nu_1); 
\end{equation*}

\noindent $\sigma^2_{\Delta GPI}$ and $\sigma^2_{\Delta GI}$ are given in the previous theorems.

\bigskip

\noindent Further,

\bigskip

$\sqrt{n}\left\{R_n(1,2) - R(1,2)\right\}$ tends to a functional Gaussian
process 
\begin{equation*}
a\,\mathcal{G}_{\Delta GPI}\, -\, b\,\mathcal{G}_{\Delta GI}
\end{equation*}

\noindent of variance-covariance function

\begin{equation*}
\sigma^2_R = a^2\,\sigma^2_{\Delta GPI} + b^2\,\sigma^2_{\Delta GI} -
\,2\,a\,b\,\sigma_{\Delta GPI,\Delta GI}.
\end{equation*}
\end{theorem}

\bigskip

\begin{proof} 
By Theorem \ref{theodeltaGI} and Theorem \ref{theodeltaGPI}, it is clear
that from Van der Vaart and Wellner (\cite{vaart}, p. $81$), the bivariate random variable 

\begin{equation*}
\left(\sqrt{n}\left(\Delta J_n(1,2) - \Delta J(1,2) \right) , \sqrt{n}%
\left(\Delta GI_n(1,2) - \Delta GI(1,2) \right) \right) 
\end{equation*}

\noindent is asymptotically Gaussian $\left(\mathcal{G}_{\Delta GPI}, 
\mathcal{G}_{\Delta GI}\right)$ with

\begin{equation*}
\sigma_{\Delta GPI,\Delta GI} = \mathbb{E}\left(\mathcal{G}_{\Delta GPI}\, 
\mathcal{G}_{\Delta GI}\right).
\end{equation*}

\noindent But let's recall that

\begin{equation*}
\mathcal{G}_{\Delta GPI} = \mathbb{G}_{(2)}\left(F^{\ast}_{J}\right) +
\beta^\ast(\nu) 
\end{equation*}

\noindent and

\begin{equation*}
\mathcal{G}_{\Delta GI} = \mathbb{G}_{(2)}\left(F^{\ast}\right) + \beta^\ast(L)
- \mathbb{G}_{(2)}\left(\tilde{F}^{\ast}\right),
\end{equation*}

\noindent then

\begin{equation*}
\sigma_{\Delta GPI,\Delta GI} = \mathbb{E}\left(\mathcal{G}_{\Delta GPI} 
\mathcal{G}_{\Delta GI} \right);
\end{equation*}

\noindent by expanding this, we find the following expression

\begin{equation*}
\sigma_{\Delta GPI,\Delta GI} =\Gamma^\ast\left(F^\ast_J,F^\ast\right) +
\Gamma^\ast\left(F^\ast_J,\beta^\ast_L\right) -
\Gamma^\ast\left(F^\ast_J,\tilde{F}^\ast\right) 
\end{equation*}
\begin{equation*}
+ \Gamma^\ast\left(F^\ast,\beta^\ast_{\nu}\right) +
\Gamma^\ast\left(\beta^\ast_{L},\beta^\ast_{\nu}\right) -
\Gamma^\ast\left(\tilde{F}^\ast,\beta^\ast_{\nu}\right), 
\end{equation*}

\noindent and we can obtain the complete expression for each covariance by using the same procedure as in the theorems \ref{theodeltaGI} and \ref{theodeltaGPI}.
\end{proof}

\bigskip


\section{Final Comments} \label{App}

\noindent We have shown hat the approach we used here, once set up, leads to powerful asymptotic laws. Besides, the construction we use allow to couple the results on the Gini index with results on aby other index as long as they are expressed in the current frame. We will not need to begin from scratch.\\

\noindent However, the variances and co-variance may have not simple forms. But this is not a major concern in the modern time of powerful computers. For example the variance and co-variance and co-variance given here may easily be performed with the free software of R.\\

\noindent To avoid to make more lengthy the paper, we decided to prepare and publish, in a very near future, papers devoted to computational methods and simulations in which we will share the codes and papers with focus on data analysis.\\


\end{document}